\documentclass{ws-procs975x65}

\def\3nab{\tilde{\nabla}}

\def\be {\begin{equation}}
\def\ee {\end{equation}}
\def\ba {\begin{eqnarray}}
\def\ea {\end{eqnarray}}
\newcommand{\bra}[1]{\left(#1\right)}

\newcommand{\brac}[1]{\left\{#1\right\}}

\newcommand{\E}{{\mathcal E}}
\renewcommand{\H}{{\mathcal H}}
\newcommand{\barray}{\begin{array}}
\newcommand{\earray}{\end{array}}
\newcommand{\e}{e}

\newcommand{\udot}{{\mathcal A}}

\begin{document}
\title{Constructing black hole entropy from gravitational collapse}
\author{Acquaviva, Giovanni $^1$, Ellis, George F. R. $^2$, Goswami, Rituparno $^3$ and Hamid, Aymen I. M.$^3$ $^*$}

\address{$^1$Department of Mathematical Sciences, University of Zululand, Private Bag X1001, Kwa-Dlangezwa 3886, South Africa.\\ $^2$Department of Mathematics and Applied Mathematics and ACGC, University of Cape Town,
Cape Town, Western Cape, South Africa.
\\$^3$Astrophysics and Cosmology Research Unit, School of Mathematics, Statistics and Computer Science, University of KwaZulu-Natal, Private Bag X54001, Durban 4000, South Africa.\\
$^1$E-mail:acquavivag@unizulu.ac.za,$^2$ george.ellis@uct.ac.zageorge.ellis@uct.ac.za$^{3}$Goswami@ukzn.ac.za $^*$aymanimh@gmail.com\footnote{Permanent address: Physics Department, Faculty of Science, University of Khartoum, Sudan} }

\begin{abstract}
Based on a recent proposal for the gravitational entropy of free gravitational fields, we investigate the thermodynamic properties of black hole formation through gravitational collapse in the framework of the semitetrad 1+1+2 covariant formalism.  In the simplest case of an Oppenheimer-Snyder-Datt collapse we prove that the change in gravitational entropy outside a collapsing body is related to the variation of the surface area of the body itself, even before the formation of horizons.  As a result, we are able to relate the Bekenstein-Hawking entropy of the black hole endstate to the variation of the vacuum gravitational entropy outside the collapsing body.
\end{abstract}


\bodymatter


\section{Introduction}
Black hole entropy in the case of eternal black holes (the maximally extended Schwarzschild vacuum solution) is a very well understood subject since the pioneering work of Bekenstein \cite{Bek73} and of Bardeen, Carter, and Hawking \cite{BarCarHaw73,Haw75}, see Ref.~\refcite{Pag05}  for a review. However astrophysical black holes form in a dynamic way. Entropy is not so well understood in that context. \\
In the context of astrophysical formation of  black holes, a key question arises.  We know that astrophysical black holes are not eternal in the past:  they are created by the continual gravitational collapse of massive stars. Therefore the question is,
\begin{quote}
\textsc{Question:} {\em Should black hole entropy be only a property of the black hole event  horizon, manifesting suddenly as the horizon forms, or should  it be an artefact of a time varying gravitational field due to gravitational collapse, 
with gravitational entropy changing smoothly from initial values to the canonical value $S_{BH} = A/4$ as an event horizon comes into being when the stellar surface area $r$ crosses the value $r = 2M$?}
\end{quote}
We use here the measure of gravitational entropy proposed in Ref.~\refcite{Clifton:2013dha}, that uses the square root of the Bel-Robinson tensor, which was shown to be unique for spacetimes which are of Petrov type D or N. This measure of gravitational entropy for free gravitational field has all the important requirements that a measure of entropy should have. It is strictly non-negative and vanishes only for conformally flat spacetimes where the Weyl tensor is zero. It measures the local anisotropy of the free gravitational field and increases monotonically as structures forms in the early universe. Most importantly, this measure reproduces the Bekenstein-Hawking entropy of a black hole: the famous theorem states that the black hole entropy at any time slice is proportional to the surface area of the black hole, which is the 2-dimensional intersection of the black hole horizon and the constant time slice \cite{Bek73,Haw75,BarCarHaw73}. Through this definition of black hole entropy, one can naturally develop the concepts of black hole thermodynamics in both classical and semiclassical regimes, leading to quantum particle creations and Hawking radiation\cite{Gibbons:1977mu}. To investigate the question stated above, in the light of the gravitational entropy proposal of Ref.~\refcite{Clifton:2013dha}, in this paper we consider the simplest example of black hole formation by Oppenheimer-Snyder-Datt \cite{osd,datt} collapse, which describes the gravitational collapse of a spherical dustlike star immersed in a Schwarzschild vacuum. Since the exterior of the star is of Petrov type D, we can uniquely determine the entropy of the free gravitational field for a static observer even when no event horizon exists. We explicitly prove that the Bekenstein-Hawking entropy of the black hole, which is formed after an infinite time for the static observer, can be linked to the net monotonic increase in the entropy of the free gravitational field during this dynamic gravitational  collapse. This result 
relates the time varying gravitational field during the continuous gravitational collapse 
to the thermodynamic property of the final state, the black hole, where gravitational entropy is well understood \cite{Pag05}.
\section{Semi-Tetrad covariant formalisms for LRS-II}
We use the 1+1+2 formalism which is an extension to the 1+3 (where there is preferred timelike vector $u^a$) formalism \cite{EllisCovariant}\cite{Ellisbook}. In this formalism, the choice of a second preferred vector along the spatial direction $e^a$ orthogonal to $u^a$ produces another split of the spacetime: this allows any 3-vector to be irreducibly split into a scalar, which is the part of the vector parallel to $\e^a$, and a vector, lying in the 2-surface orthogonal to $\e^a$. 
We apply this formalism to {\it locally rotationally symmetric} (LRS) spacetime \cite{EllisLRS}.  
The variables that uniquely describe an LRS spacetime are
$\brac{\udot, \Theta,\phi, \xi, \Sigma,\Omega, \E, \H, \mu, p, \Pi, Q }$.  Within the LRS class, the LRS-II admits spherically symmetric solutions, is free of rotation and is described by the variables $\brac{\udot, \Theta,\phi, \Sigma,\E, \mu, p, \Pi, Q }$, since $\Omega$, $ \xi $ and $ \H $ all vanish. These spacetimes include Schwarzschild, Robertson-Walker, Lema\^{i}tre-Tolman-Bondi (LTB), and Kottler spacetimes.
\section{Thermodynamics of a gravitational field}
The second law of thermodynamics can be written using the thermodynamically motivated proposal for the gravitational entropy measure \cite{Clifton:2013dha} as
\begin{equation}\label{dels}
 \delta s_{grav}=\frac{\delta\bra{\mu_{grav} v}}{T_{grav}}.
\end{equation}
where $s_{grav}$, $T_{grav}$ and $\mu_{grav}$ represent the gravitational entropy density, effective temperature, and the energy density of the free gravitational field respectively. $v$ is the spatial volume. We integrate eq.(\ref{dels}) over a spacelike hypersurface. The last ingredient needed is a definition for the temperature of the gravitational field.
For the gravitational temperature we follow the proposal of Ref.~\refcite{Clifton:2013dha} which 
can be represented in the 1+1+2 decomposition as $T_{grav}=\frac{|\udot+\frac{1}{3}\Theta+\Sigma|}{2\pi}$.
\subsection{Gravitational entropy and structure formation}
We have already stated before, the square root of Bel-Robinson tensor being the measure of gravitational entropy,  enables structure formation naturally as the entropy increases as the structure (or inhomogeneities) forms in the universe. We drive the following relation between the energy density of the gravitational field and the Misner-Sharp mass 
$\mu_{grav}=\frac{\alpha K^{3/2}}{6\pi\phi}\left|\bra{\hat{{\cal M}}_{ms}-\frac32\phi{\cal M}_{ms}}\right|\;$.
Now from LRS-II field equations, we can easily see that for a homogeneous distribution of perfect fluid with $\hat{\mu}=\hat{p}=0$ we have $\hat{{\cal M}}_{ms}-\frac32\phi{\cal M}_{ms}=0$ on every constant time slice and hence $ \delta s_{grav}=0$. However as discussed in Ref.~\refcite{Clifton:2013dha}, if we start with an inhomogeneous distribution of collapsing matter (as it happens during structure formation), we have $\hat{{\cal M}}_{ms}-\frac32\phi{\cal M}_{ms}\ne0$. This will then make $d\mu_{grav}>0$ and hence $dS_{grav}>0$. Thus the thermodynamics of free gravity naturally favours structure formation, in contrast with the thermodynamics of standard matter that favours dispersion. In the light of above discussion we can predict that the vacuum gravitational entropy outside a collapsing star (integrated over each constant time slice) will increase with time, favouring the process of continual gravitational collapse. This we prove explicitly in the next section.
\section{Gravitational entropy of the vacuum around a collapsing star}
Having all the ingredients at hand, we want to look now at the variation of gravitational entropy outside a body which is collapsing to form a black hole in the simplest scenario, the Oppenheimer-Snyder-Datt \cite{osd,datt} dust collapse model, represented schematically in Fig. \ref{collapse}. 
The interior of the collapsing star is described by a FLRW spacetime matched with the exterior vacuum solution represented by Schwarzschild spacetime. Though in general the entropy of the spacetime has contributions coming both from the matter and the gravitational field, but in the interior spacetime the matter entropy only contributes, since $\E=0$ for FLRW spacetimes while in the exterior vacuum only the gravitational part of the entropy does not vanish. 
Being interested in the variation of gravitational entropy from the point of view of an external static observer, we choose to integrate eq.(\ref{dels}) over a spacelike hypersurface {\it outside} the collapsing body.  Based on the assumptions stated above, in this scenario we are able to show the following:
\begin{prop}\label{propos1}
 The increase in the instantaneous gravitational entropy outside a collapsing star during a given interval of time is proportional to the change in the surface area of the star during that interval.
\end{prop}
\begin{proof}
Outside the collapsing star the spacetime is Schwarzschild.  Therefore, by using the gravitational energy density and the temperature definition in eq.(\ref{dels}) and integrating over a 3-volume of the exterior region at fixed time, the total entropy at a given time can be expressed as
\be\label{entropyint}
S_{grav}\equiv\int_{\sigma}\, \delta s_{grav} = \pi \alpha \int_{R(\tau_0)}^{\infty}\frac{|\E|}{\udot}\ \frac{\bar{R}^2}{\sqrt{1-\frac{2m}{\bar{R}}}} d\bar{R},
\ee
where we have used $v=u^a\eta_{abcd}dx^b dx^c dx^d$, and the timelike vector $u^a$ is given by
$u^a=\bra{1/\sqrt{|1-\frac{2m}{R}|},0,0,0}$, $R(\tau_0)$ is the radius of the collapsing star at time $\tau_0$ and $m$ is the total mass of the star.
We know that for Schwarzschild spacetime $\E$ and $\udot$ are given by \cite{Betschart:2004uu} $|\E|=\frac{2m}{R^3},\quad \udot=\frac{m}{R^2}\bra{1-\frac{2m}{R}}^{-1/2}$.
Using the last two expressions in eq.(\ref{entropyint}), the gravitational entropy is then given by $S_{grav}=2 \pi \alpha \int_{R(\tau_0)}^{\infty}\bar{R}d\bar{R}\ .$
This is an infinite quantity, (although this can be made finite by using the idea of {\em Finite-Infinity} for a realistic astrophysical star for which spacetime is almost Minkowski at a distance of one light year). However if we calculate the change in the gravitational entropy in a time interval $(\tau-\tau_0)$ we obtain
\be
\delta S_{grav}|_{(\tau-\tau_0)}=\frac{\alpha}{4}\Big(A(\tau_0)-A(\tau)\Big),\label{deltas}
\ee
where $A(\tau)$ is the surface area of the star at any time $\tau>\tau_0$. 
\end{proof}
\begin{figure}[ht!]
 \begin{center}
 \includegraphics[width=0.8\textwidth]{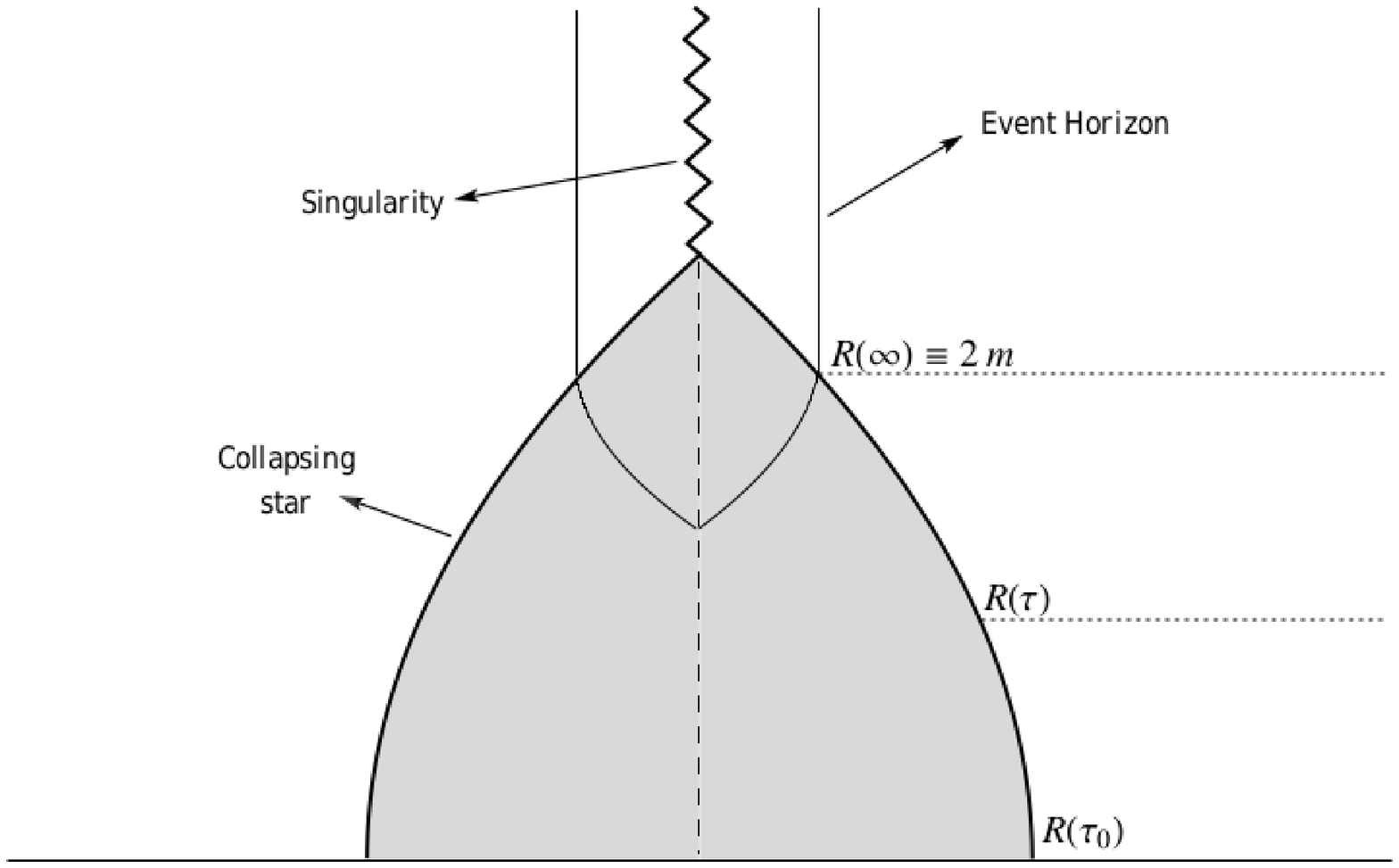}
 \caption{\label{collapse} Oppenheimer-Snyder dust collapse of a star (shaded).  In the reference frame of a static external observer, the crossing of the star's surface with the horizon at radius $2m$ occurs at $\tau\rightarrow\infty.$}
\end{center}
\end{figure}
The value of the parameter $\alpha$ can be constrained if we consider the variation of entropy between a configuration with $R(\tau_{\epsilon})=2m+\epsilon$ (with $\epsilon\ll2m$) and the asymptotic black hole state with $R(\infty)=2m$.  The time elapsed between these spatially neighboring states is actually infinite, because the formation of a black hole as a result of the collapse from the point of view of a static external observer is a process that takes an infinite amount of time.  Eq.(\ref{deltas}) gives $ \delta S_{grav}|_{(\infty-\tau_{\epsilon})}=\frac{\alpha}{4}\Big(A(\tau_{\epsilon})-A_H\Big) \simeq\alpha\, 4\pi\, m\, \epsilon$
where $\simeq$ means that we are neglecting higher orders in $\epsilon$.  The energy/mass supplied by this final stage of collapse to form the black hole is $dU\simeq\epsilon/2$, so that from the above we can be rewrite $dS\simeq\alpha\, (8\pi m)\ dU$.
The term in round brackets is the Hawking temperature $T_H=(8\pi m)^{-1}$ and hence, if $\alpha=1$, we recover the second law of black hole thermodynamics \cite{Bardeen:1973gs}.  The same value for $\alpha$ was found in Ref.~\refcite{Clifton:2013dha} by calculating the instantaneous gravitational entropy of a black hole and comparing the result with the known Bekenstein-Hawking value.

\subsection{Building up the Bekenstein-Hawking entropy}

As a consequence of Proposition \ref{propos1}, we write the following corollary

\begin{corollary}
 The Bekenstein-Hawking entropy of a black hole, formed as an endstate of a spherically symmetric collapse of a massive star with Schwarzschild spacetime as the exterior, is the difference between one fourth of the initial area of the collapsing star and the net increase in the vacuum entropy in infinite collapsing time.
\end{corollary}
\begin{proof}
In the reference frame of an external static observer, the crossing between the collapsing star's surface and the horizion (and hence the formation of the black hole) will take an infinite time.  Assuming that for the asymptotic black hole endstate the Bekenstein-Hawking relation $S_{BH}=\frac{1}{4}A_H$ holds, where $A_H$ is the area of the event horizon, then from Proposition \ref{propos1} with $\alpha=1$ we have $S_{BH}=\frac{1}{4}A(\tau_0)- \delta S_{grav}|_{(\infty-\tau_0)}\ $.
\end{proof}
No structure can form spontaneously. But in fact order does indeed spontaneously form on large scales as the universe expands - an apparent contradiction with the second law \cite{Ellis:95}. In order to resolve this, one needs a good definition of gravitational entropy.
The definition given in Ref.~\refcite{Clifton:2013dha}, where (following Penrose' suggestions) gravitational entropy is based in the properties of the Weyl tensor, resolves this issue as far as the growth of perturbations in the expanding universe, due to gravitational attraction, is concerned (see equations (54) and (55) in Ref.~\refcite{Clifton:2013dha}). The present paper has shown that that initial growth of gravitational entropy, taking place in conjunction with the initial formation of structure in the expanding universe, can be smoothly joined on to the formation of black holes. The famous black hole entropy does not suddenly appear when the event horizon is formed; it grows steadily as gravitational attraction causes ever more concentrated objects to form, eventually leading to the existence of black holes with the standard gravitational entropy.

\section*{acknowledgments}
GA is thankful to the Astrophysics and Cosmology Research Unit (ACRU) at the University of KwaZulu-Natal for the kind hospitality.  AH,  RG and GFRE are supported by National Research Foundation (NRF), South Africa.

\end{document}